\documentclass[11pt]{article}

\usepackage{enumerate}
\usepackage{latexsym} 
\usepackage{theorem}
\usepackage{graphics}
\usepackage{graphicx}
\usepackage{amsmath}
\usepackage{amssymb}  
\usepackage[english]{babel} 
\usepackage{comment}
\usepackage{color}
\usepackage{ifpdf}

\usepackage{amsfonts}
\usepackage{setspace}
\usepackage{fullpage}
\usepackage{enumitem}
\usepackage{hyperref}
\usepackage{microtype}
\usepackage{enumerate}

\newtheorem{lemma}{Lemma}[section]
\newtheorem{theorem}[lemma]{Theorem}

\newtheorem{observation}[lemma]{Observation}

\newcommand{\mc}{\mathcal}
\newcommand{\sm}{\setminus}

\newcommand{\dic}{\vec{\chi}}
\newcommand{\olra}{\overleftrightarrow}

\newcommand{\dmax}{\Delta_{max}}
\newcommand{\dmin}{\Delta_{min}}

\usepackage{pgf,tikz}
\usetikzlibrary{patterns}
\usetikzlibrary{calc}
\tikzstyle{vertex}=[circle, draw, inner sep=0pt, minimum size=8pt]

\usetikzlibrary{arrows}
\usetikzlibrary{decorations.markings}

\usetikzlibrary{decorations.pathreplacing}

\tikzset{->-/.style={decoration={
  markings,
  mark=at position .5 with {\arrow{>}}},postaction={decorate}}}

\tikzstyle{vertex}=[circle,draw, top color=gray!5, 
	    bottom color=gray!30, minimum size=12pt, scale=1, inner sep=0.5pt]
\tikzstyle{arc}=[->, > = latex',  thick]
\tikzstyle{edge}=[thick, blue]

\def\centerarc[#1](#2)(#3:#4:#5)
    { \draw[#1] ($(#2)+({#5*cos(#3)},{#5*sin(#3)})$) arc (#3:#4:#5); }

\newcounter{claim}

\newenvironment{proof}[1][]%
 {\noindent {\setcounter{claim}{0}\sc proof ---
   }{#1}{}}{\hfill$\Box$\vspace{2ex}} 

\newenvironment{claim}[1][]%
{\refstepcounter{claim}\vspace{1ex}\noindent{(\it\arabic{claim}){#1}{}}\it}{\vspace{1ex}}

\newenvironment{proofclaim}[1][]%
	{\noindent {}{#1}{}}{ This proves~(\arabic{claim}).\vspace{2ex}}

\title{Four proofs of the directed Brooks' Theorem}

\author{Pierre Aboulker$^1$, Guillaume Aubian$^{1,2}$\\
\small ($1$) DIENS, \'Ecole normale sup\'erieure, CNRS, PSL University, Paris, France\\
\small ($2$) Université de Paris, CNRS, IRIF, F-75006, Paris, France.}

\begin{document}

\maketitle
\begin{abstract}
 We give four new proofs of the directed version of Brook's Theorem and an NP-completeness result. 
\end{abstract}

\section{Introduction}
A \emph{ $k$-colouring} of an undirected graph $G$ is a partition $V_1, \dots, V_k$ of $V(G)$ into $k$ independent sets. The \emph{chromatic number} of $G$, denoted  $\chi(G)$, is the least $k$ such that $G$ admits a $k$-colouring. 
The maximum degree of an undirected graph $G$ is denoted by $\Delta(G)$. It is an easy observation that for every graph $G$, $\chi(G) \leq \Delta(G)+1$. The following classical result of Brooks characterizes the (very few) graphs for which equality holds. 
\begin{theorem}[Brooks' Theorem, \cite{B41}]
A connected graph $G$ satisfies $\chi(G) = \Delta(G) + 1$ if and only if $G$ is an odd cycle or a complete graph. 
\end{theorem}

Many proofs of Brooks' Theorem have been found, and the different proofs generalize and extend in many directions. See~\cite{CR14} for a particularly nice survey on this subject. 
Brooks' Theorem has been generalised to digraphs via the notion of acyclic colouring. 
The aim of this paper is to give four new proofs of the directed version, each of them adapted from a proof of the undirected version. 
\medskip

The digraphs in this paper  have no loops or parallel arcs, but we allow cycle of length $2$ (\textit{digon}). A digraph is \textit{acyclic} if it contains no directed cycle. 
An \emph{acyclic colouring} (or \emph{dicolouring}) of a digraph $G$ is a colouring of $V(G)$ in such a way that no directed cycle is monochromatic. Equivalently, it is a partition of $G$ into acyclic induced subdigraph. 
The \emph{dichromatic number $\dic(G)$} of a digraph $G$ is the minimum number of colors in an acyclic colouring of $G$. 

The dichromatic number was first introduced by Neumann-Lara~\cite{NL82} in 1982 and was  rediscovered by Mohar~\cite{M03} 20 years later.
It is easy to see that for any undirected graph $G$, the \textit{symmetric digraph} $\olra G$ obtained from $G$ be replacing each edge by a digon satisfies $\chi(G) = \dic(\olra G)$. This simple fact permits to generalize results on the chromatic number of undirected graphs to digraphs via the dichromatic number. Such results have (recently) been found in various areas of graph colouring such as extremal graph theory~\cite{BBSS20, HK15, KS20}, algebraic graph theory~\cite{M10}, substructure forced by large dichromatic number~\cite{AAC21, ACN21, ACL19, hero,GSS20, HLNT19, S21}, list dichromatic number~\cite{BHL18, HM11}, dicolouring digraphs on surfaces~\cite{AHKR21, LM17, S19}, flow theory~\cite{H17, KV12}, links between dichromatic number and girth~\cite{HM12, S20}. 
\medskip


The maximum degree of a graph does not have a clear analogue for digraphs. 
We now introduce two ways to measure maximum degree in a digraph that make sense in the context of Brooks' Theorem.  
Let $v$ be a vertex of a digraph $G$. 
We define the \emph{maxdegree} of $v$ as $d_{max}(v)= \max(d^+(v), d^-(v))$ and the \emph{mindegree} of $v$ as $d_{min}(v)= \min(d^+(v), d^-(v))$.
We can then define the corresponding maximum degrees: $\Delta_{max}(G)= \max_{v \in V(G)}(d_{max}(v))$ and $\Delta_{min}(G)= \max_{v \in V(G)}(d_{min}(v))$. 
The following easily holds (see subsection~\ref{subsec:def} for a proof): for every digraph $G$, $\dic(G) \leq \dmin(G) +1 \leq \dmax(G) +1$. 
 

A \emph{symmetric cycle} (resp. \emph{symmetric complete graph}) is the digraph obtained from a cycle (resp. from a complete graph), by replacing each edge by a digon.

We are now ready to state the  directed version of Brooks' Theorem. 
It was first proved by Mohar in~\cite{M10}, but we discovered that the proof is incomplete, see Section~\ref{sec:lovasz} for more details. Anyway, in~\cite{HMGallai}, Harutyunyan and Mohar generalised Gallai's Theorem  (a strengthening of Brooks' Theorem for list colourings) to digraph, which gave an alternative and correct proof. 

\begin{theorem}[\cite{M10, HMGallai}]\label{brooks_max_oriented}
Let $G$ be a connected digraph, then $\dic(G) \leq \Delta_{max}(G) + 1$ and equality holds if and only if one of the following occurs:
\begin{itemize}
\item[(a)] $G$ is a directed cycle or,
\item[(b)] $G$ is a symmetric cycle of odd length or,
\item[(c)] $G$ is a symmetric complete graph on at least 4 vertices.
\end{itemize}
\end{theorem}

The next four sections are devoted to four new proofs of the directed Brooks' Theorem. 
In the last section, we show that  it is NP-complete to decide if $\dic(G) = \Delta_{min}(G) +1$, so a simple characterization of digraphs satisfying $\dic(G) = \dmin(G) +1$ is very unlikely. 


\subsection{Definitions and preliminaries} \label{subsec:def}
Let $G$ be a digraph and $v$ a vertex of $G$. 
We denote by $d_G^+(v)$ (resp.~$d_G^-(v)$) the number of out-neighbours (resp.~of in-neighboyrs) of $v$. We omit the subscript when $G$ is clear from the context.  
We denote by $N^+(v)$ (resp.~$N^-(v)$) the set of out-neighboyrs (resp.~in-neighbours) of $v$,  and by $N(v)$ the set of neighbours of $v$, that is $N(v) = N^+(v) \cup N^-(v)$. If $X$ is a set of vertices and $v \notin X$, $N_X(v) = N(v) \cap X$, $N^+_X(v)$ and $N^-_X(v)$ are defined similarly. We denote by $G[X]$ the subdigraph of $G$ induced by $X$. 
A digraph is \emph{$k$-regular} if for every vertex $v$, $d^+(v)=d^-(v)=k$. 
\medskip

We denote by $\mc B_{1}$ the set of directed cycles, $\mc B_{2}$ the set of symmetric odd cycles and, for $k \geq 3$, $\mc B_{k} = \{\olra K_{k+1}\}$ where $\olra K_{k+1}$ is the symmetric complete graph on $k + 1$ vertices.
Observe that the directed version of Brooks Theorem is equivalent to the following statement: \emph{A digraph $G$ has dichromatic number at most $\dmax(G) +1$ and equality occurs if and only if $G$ contains a connected component isomorphic to a member of $\mc B_{\dmax(G)}$}. 
We sometimes call the members of $\mc B_k$ \emph{exceptions}.
\medskip

Given a digraph $G$ and an ordering $(v_1, \dots, v_n)$ of its vertices, 
to \textit{colour greedily}  $G$ is to colour $v_1, \dots, v_n$ in this order by giving to $v_i$ the minimum between the smallest colour not used in $N^+(V) \cap \{v_1, \dots, v_{i-1}\}$ and the smallest colour not used in $N^-(V) \cap \{v_1, \dots, v_{i-1}\}$. 
It is easy to see that any ordering leads to an acyclic colouring with at most $\Delta_{min}(G)+1$ colours. And since we clearly have $\dmin(G) \leq \dmax(G)$, we have:
$$\dic(G) \leq \Delta_{min}(G) +1 \leq \dmax(G) +1$$

Given a digraph $G$, we define by $\tilde G$ its underlying graph and we say that $G$ is \textit{connected} if its underlying graph is connected. 
The following easy lemma will be used in the four proofs of the directed Brooks' Theorem. Note that it does not hold if one replaces $\dmax(G)$ by $\dmin(G)$, implicit examples are given in Section \ref{sec:dmin}. 

\begin{lemma}\label{lem:reg}
If $G$ is a connected non-regular digraph, then  $\dic(G) \leq \Delta_{max}(G)$.
\end{lemma}

\begin{proof}
Since $G$ is non-regular, it has a vertex $u_1$ such that $d_{min}(u_1) < \Delta_{max}(G)$. Let $u_1, \dots, u_n$ be a vertex ordering output by a BFS on $\tilde G$ starting at $u_1$. By greedily colouring $G$ with respect to the ordering $u_n, \dots, u_1$, we get a colouring with at most  $\Delta_{max}(G)$ colours.
\end{proof}

If $\Delta_{max}(G) = 1$, then every vertex has at most one in-neighbour and at most one out-neighbour so $G$ is a directed cycle or a path. Hence, $\dic(G) = 2$ if and only if $G$ is a directed cycle. This proves Theorem \ref{brooks_max_oriented} for $\Delta_{max}(G) = 1$. So we only need to prove the directed Brooks' Theorem for digraphs with $\dmax(G) \geq 2$,  and we have the base case when we want to proceed by induction on the value of $\dmax(G)$. 


\section{Lov\'asz' proof: greedy colouring}\label{sec:lovasz}

In this section, we adapt the proof of Brooks' Theorem given by Lov\'asz in~\cite{L75}. 
The idea is the following: when we greedily colour the vertices of a connected digraph $G$ using the reverse order output by a BFS of $\tilde G$, each vertex except (possibly) the last one receives a colour from $\{1, \dots, \Delta_{max}(G)\}$. Indeed, the fact that $G$ is connected ensures that each vertex (except possibly the last one) has at most $\dmax(G) -1$ in-neighbours or out-neighbours already coloured. 
The goal of the proof is then to find an ordering of the vertices such that the last vertex can also be coloured with a colour from $\{1, \dots, \Delta_{max}(G)\}$. 

The first version of the directed Brooks' Theorem appeared in~\cite{M10} and the  given proof is based on Lov\'asz' idea, but appears to be incomplete.  
To explain why, let us dive a little deeper into the proof. 
The goal is to find a vertex $v$ that has two out- (or two in-) neighbours $v_1$, $v_2$ such that $v_1$ and $v_2$ are not linked by a digon and such that $G \sm \{v_1, v_2\}$ is connected. 
You can then choose an ordering of the vertices that starts with $v_1$ and $v_2$ and continue with the reverse order output by a BFS of $\tilde G$ starting at $v$ (so the ordering ends with $v$). A greedy colouring give colour $1$ to $v_1$ and $v_2$, and thus there will be an available colour from $\{1, \dots, \dmax(G)\}$ to colour $v$ (the last vertex of the ordering). 
In~\cite{M10}, a vertex $v$ with two in- or two out-neighbours $v_1$ and $v_2$ not linked by a digon is found, but the fact that $G \sm \{v_1, v_2\}$ is connected is not checked, and reveals to be non-trivial to prove.  
We now give a full proof based on this idea.

\begin{theorem}
A connected digraph $G$ has dichromatic number at most $\dmax(G) +1$ and equality occurs if and only it is a member of $\mc B_{\dmax(G)}$. 
\end{theorem}

\begin{proof}
Let $G$ be a counter-example, that is $G$ is connected, $\dic(G) = \Delta_{max}(G) +1$ and $G$ is not a member of $\mc B_{\dmax(G)}$. Set $k = \Delta_{max}(G) \geq 2$ and recall that $\tilde{G}$ denotes the underlying graph of $G$.
By Lemma~\ref{lem:reg}, $G$ is $k$-regular.

\begin{claim}\label{claim:2con}
 $\tilde{G}$ is $2$-connected
\end{claim}

\begin{proofclaim}
Assume for contradiction that $\tilde G$ has a cutvertex $u$ and let $C_1$ be a connected component of $G - u$, and $C_2$ the union of the other connected components.  
Set $G_{i} = G[C_{i} \cup \{u\}]$ for $i=1,2$. By Lemma~\ref{lem:reg}, $G_1$ and $G_2$ are $k$-dicolourable. 
Up to permuting colours, we may assume that the $k$-dicolourings of $G_1$ and $G_2$ agree on $u$, which give a $k$-dicolouring of $G$, a contradiction. 
\end{proofclaim}

\begin{claim}\label{claim:2edgeCut}
  $\tilde G$ has no edge-cut of size $2$.
\end{claim}

\begin{proofclaim} 
Assume by contradiction that $G$ has an edge cutset $\{e_1, e_2\}$. 
Let $G_{1}$ and $G_2$ be the two connected components of $G - \{e_1, e_2\}$. 
Both $G_1$ and $G_2$ are $k$-colourable by Lemma \ref{lem:reg}. 
A $k$-coloring of $G_1$ and $G_2$ give a $k$-colouring of $G$ as soon as the extremities of $e_1$ and $e_2$ use at least two distinct colours. Permuting colours in $G_1$ if necessary, we get a $k$-colouring of $G$. 
\end{proofclaim}

\begin{claim}\label{claim:2cut}
If $\{u,v\} \subseteq V(G)$ is a cutset of $\tilde{G}$, then $\{u,v\}$ is a stable set. 
\end{claim}

\begin{proofclaim}
Let $\{u,v\} \subseteq V(G)$ be a cutset of $\tilde{G}$ and assume for contradiction and without loss of generality, that $uv$ is an arc of $G$. 
Let $C_{1}$ be a connected component of $\tilde G \sm \{u,v\}$ and $C_2$ the union of the other connected components.  Set $G_{i} = G[C_{i} \cup \{u, v\}]$ for $i=1,2$.

Since $\tilde{G}$ is $2$-connected, both $u$ and $v$ have some neighbours in both $C_1$ and $C_2$ and thus $G_1$ and $G_2$ are $k$-dicolourable by Lemma~\ref{lem:reg}.
 If both $G_1$ and $G_2$ admit a $k$-dicolouring in which $u$ and $v$ receive distinct (resp. same) colours, then we get a $k$-dicolouring of $G$, a contradiction (because no induced cycle can intersect both $C_1$ and $C_2$). 
 So we may assume without loss of generality that $u$ and $v$ receive the same colour (resp. distinct colours) in every $k$-dicolouring of $G_1$ (resp. in every $k$-dicolouring of $G_2$). 

If $u$ has an out-neighbour in $C_2$, then $d_{G_1}^+(u) \leq k-1$. We can $k$-dicolour $G_1 - \{u\}$, and extend the $k$-dicolouring to $u$ with a colour not appearing in the out-neighbourhood of $u$, so in particular distinct from the colour of $v$, a contradiction.
So $u$ has no out-neighbour in $C_2$ and similarly, $v$ has no in-neighbour in $C_2$. 

Suppose $u$ has in-degree at least $2$ in $G_{2}$. Then $d_{G_1}^-(u) \leq k-2$ and thus we can $k$-dicolour $G_1 - \{u\}$ and extend this dicolouring to $G_1$ by giving to $u$ a colour not used in its in-neighbour and distinct from $v$, a contradiction. 
So $u$ has exactly one in-neighbour in $G_2$, and similarly $v$ has exactly one out-neighbour in $G_2$ which gives us an edge cutset of size $2$, a contradiction with (\ref{claim:2edgeCut}).
\end{proofclaim}

\begin{claim}\label{claim:lovasz}
Let $x$ be a vertex of $G$ and $u$ and $v$ two out-neighbours of $x$. Then either $\{u,v\}$ induces a digon, or $\{u,v\}$ is a cutset. Same holds if $u$ and $v$ are in-neighbours of $x$. 
\end{claim}

\begin{proofclaim}
Assume for contradiction that $\{u, v\}$ does not induce a digon and is not a cutset of $G$. 
Let $G' = G - \{u, v\}$ and $\tilde{G}'$  the underlying graph of $G'$. Since $\tilde{G}'$ is connected, there is a BFS ordering $(x = u_{1}, u_{2}, \dots, u_{n-2})$ of $\tilde{G}'$. Set $u_{n-1} = u$ and $u_{n} = v$. 
We now greedily dicolour $G'$ with respect to the order $(u_n, u_{n-1}, \dots, u_1)$.  Since $G[\{u_{n}, u_{n-1}\}]$ is not a digon, $u_{n}$ and  $u_{n-1}$ both receive colour $1$. For $i=n-2, \dots 2$, $u_i$ has at least one neighbour in $G[\{u_1, \dots u_{i-1}\}]$, and thus $u_i$ has at most $k-1$ in- or out-neighbours in $G[u_n, \dots, u_i]$ and hence we can assign a colour from $\{1, \dots, k\}$ to it.  
Finally, since $u_n$ and $u_{n-1}$ receive colour $1$ and are both in the out-neighbourhood of $u_1$, the out-neighbourhood of $u_1$ is coloured with at most $k-1$ distinct colours and thus $u_1$ receive a colour from $\{1, \dots, k\}$, a contradiction. 
The proof is the same when $u$ and $v$ are in-neighbours of $x$.
\end{proofclaim}
\medskip

Observe that $G$ cannot be a symmetric digraph because of the undirected Brook's Theorem. So there exists $u, v \in V(G)$ such that $uv \in A(G)$ and $vu \notin A(G)$. 
By (\ref{claim:2cut}), $\{u, v\}$ is not a cutset.

\begin{claim}\label{neighboursofuarecut}
For every $a \in N^+(u) \setminus \{v\}$, $\{a, v\}$ is a cutset. 
\end{claim}

\begin{proofclaim}
Suppose $\{a, v\}$ is not a cutset. By (\ref{claim:lovasz}) $\{a,v\}$ induces a digon and thus $u$ and $v$ are in-neighbours of $a$. But $\{u,v\}$ is not a cutset by (\ref{claim:2cut}) and does not induce a digon, a contradiction to (\ref{claim:lovasz}).
\end{proofclaim}
\medskip

Let $H = G - v$ and let $a \in N^+(u) \setminus \{v\}$. 
By (\ref{neighboursofuarecut}) $a$ is a cutvertex of $H$, so $H$ has at least two blocks (where a \textit{block} is a maximal 2-connected subgraph of $\tilde G$).
Since $\tilde G$ is $2$-connected, $v$ has a neighbour in each leaf block of the block decomposition of $\tilde H$. 

We now break the proof into two parts with respect to the value of $k$. 
Suppose first that $k = 2$. 
If the two out-neighbours (resp. the two in-neighbours) of $v$ belong to distinct blocks of $\tilde H$, then $N^+(v)$ does not induce a digon, nor a cutset of $G$, a contradiction to~(\ref{claim:lovasz}). 
Hence $N^+(v)$ is included in a leaf block of $H$ and  $N^-(u)$ in another one.
Now, dicolour $H$ with $2$ colours (it is possible by Lemma~\ref{lem:reg}). 
Let $w$ be a cutvertex of $H$ separating the leaf blocks containing the neighbours of $v$. Observe that every cycle containing $v$ must go through $w$. Hence we can extend the $2$-dicolouring of $H$ by giving to $v$ a colour distinct from the one received by $w$  to get a $2$-dicolouring of $G$, a contradiction.

Assume now that $k \geq 3$. So there exists $b \in N^+(u) \setminus \{a,v\}$. By (\ref{neighboursofuarecut}), both $a$ and $b$ are cutvertices of $H$. Since $uv \in A(G)$,  $u$ is not a cutvertex of $H$ by (\ref{claim:2cut}). 
Let $U$ be the block of $H$ containing $u$ (which is unique because $u$ is not a cutvertex of $H$). Since $u$ sees both $a$ and $b$, $U$ is not a leaf block of $H$. 
Let $U_{1}$ and $U_{2}$ be two distinct leaf blocks of $H$. Since $\tilde{G}$ is $2$-connected, $v$ must have neighbours in $U_1$ and $U_2$. Let $u_1 \in U_1$ and $u_2 \in U_2$ be two neighbours of $v$. 
So $u$, $u_1$, $u_2$ are in pairwise distinct blocks of $H$ which implies that for every $\{x,y\} \subseteq \{u,u_1,u_2\}$, $\{x,y\}$ does not induced a digon and is not a cutset of $\tilde{G}$. 
Now, since $u$, $u_1$, $u_2$ are neighbours of $v$, two of them are included in the in-neighbourhood or in the out-neighbourhood of $v$, a contradiction to (\ref{claim:lovasz}).
\end{proof}


\section{Acyclic subdigraph and induction}

The proof of this section is an adaptation of a proof of Rabern~\cite{R14}, see also Section 3 of \cite{CR14}. 
Here is a sketch of the proof. 
Let $G$ be a digraph with $\dmax(G) = k$. We do an induction on $k$. 
We first choose a maximal induced acyclic subdigraph $M$ of $G$ and prove that $G-M$ must have dichromatic number $k-1$ and thus must contain a connected component $T$ isomorphic to a member of $\mc B_{k-1}$ by induction. We then show that a $k$-dicolouring of $G-T$ can be extended to $G$.

\begin{theorem}
Let $G$ a digraph such that $\dic(G) = \Delta_{max}(G) + 1$. Then $G$ contains a connected component isomorphic to a member of $\mc B_{\Delta_{max}(G)}$. 
\end{theorem}
\begin{proof}
The theorem is true for digraphs $G$ with $\dmax(G) = 1$. Let $k \geq 2$ and assume the theorem holds for digraph with maximum maxdegree at most $k-1$. By mean of contradiction, assume there exists a digraph $G$ with $\dmax(G) = k$ violating the theorem. 
We choose such a $G$ with minimum number of vertices. 
By Lemma~\ref{lem:reg}, $G$ is $k$-regular. 

We now prove two technical claims. 

\begin{claim}
If $k \geq 3$, $G$ cannot contain $\olra{K}_{k+1}$ less an arc, or less a digon, as an induced subdigraph. 
\end{claim}

\begin{proofclaim}
Suppose $G$ contains a subdigraph $K$ isomorphic to $\olra{K}_{k + 1}$ less a digon $\{uv,vu\}$.  
Observe that $u$ and $v$ both have exactly one in-neighbour and one out-neighbour outside of $K$, and that all other vertices of $K$ have no neighbour outside of $K$. Now, by Lemma~\ref{lem:reg}, $G - K$ can be $k$-dicoloured and we can extend this $k$-dicolouring to $G$ as follows: at most one colour is forbidden for $u$ and one for $v$, hence, since $k \geq 3$, we can give the same colour to $u$ and $v$, and then assign the $k-1$ remaining colours to $V(K)\sm \{u,v\}$. We thus get a $k$-dicolouring of $G$, a contradiction.  The same reasoning holds when an arc is missing instead of a digon. 
\end{proofclaim}

\begin{claim}
If $k = 2$, $G$ cannot contain a symmetric odd cycle less an arc, or less a digon, as an induced subdigraph. 
\end{claim}

\begin{proofclaim}

Let $\ell \geq 1$.
Assume for contradiction that $G$ contains a subdigraph $C$ isomorphic to $\overleftrightarrow{C}_{2\ell + 1}$ less an arc $uv$. Let us consider a $2$-dicolouring of $G - C$ and assume without loss of generality that the out-neighbour of $u$ not in $C$ is coloured $1$. We can colour $u$ and $v$ with colour $2$, and greedily dicolour $C - u - v$ to obtain a $2$-dicolouring of $G$, a contradiction.

Suppose now that $G$ contains a subdigraph $C$ isomorphic to $\overleftrightarrow{C}_{2l + 1}$ less a digon $\{uv,vu\}$. 
Let us name $F = (G - (C - \{u,v\})) / uv$. Either $F$ is $2$-dicolourable, in which case there exists a $2$-dicolouring of $G - \{C-\{u,v\}\}$ in which $u$ and $v$ receive the same colour and we can extend this dicolouring to $C$ or, as $\Delta_{max}(F) \leq 2$ and $|V(F)| < |V(G)|$, $F$ is a symmetric odd cycle, which implies $G$ is a symmetric odd cycle as well, a contradiction.
\end{proofclaim}

Let $M$ be a maximal directed acyclic subdigraph of $G$. By maximality of $M$, every vertex in $G-M$ must have at least one in-neighbour and one out-neighbour in $M$, so $\Delta_{max}(G - M) \leq k-1$. 
Moreover, $\dic(G - M) = k$,  as otherwise we could $(k - 1)$-dicolour $G - M$ and use a $k^{th}$ colour for $M$. So $G-M$ has a connected component $T$ isomorphic to a member of $\mathcal B_{k-1}$ by induction. 

Suppose first that there exists $u \in V(T)$ whose in-neighbour $x$ and out-neighbour $y$ in $G - T$ are distinct. 
Let $H = G - T$ to which is added the arc $xy$ if   $xy \notin A(G)$. 
Observe that $\Delta_{max}(H) \leq k$. 
Then $H$ does not contain any element of $\mathcal B_{k}$ (as $G$ does not contain an element of $\mathcal B_{k}$ less an arc) which, by minimality of $G$, implies that $H$ is $k$-dicolourable. Thus there is a $k$-dicolouring of $G - T$ with no monochromatic path from $y$ to $x$.

We are now going to show that such a dicolouring can be extended to $T$. We break the proof into two parts with respect to the value of $k$. 

Assume first that $k \geq 3$. Then $T$ induces $\olra K_{k}$. 
Observe that each vertex of $T$ has precisely one in-neighbour and one out-neighbour outside of $T$. So we can greedily extend the $k$-dicolouring of $G-T$ to $G-u$. We can now greedily extend this dicolouring to $u$. This is possible because there is no monochromatic path from $y$ to $x$ in $G-T$. 


Assume now that $k=2$. Then $T$ induces a directed cycle. 
If $\cup_{v \in T} N(v) \setminus V(T)$  is monochromatic of colour $c$, we can assign colour $c$ to $u$ and the other colour to vertices of $T-\{u\}$ to obtain a proper $2$-dicolouring of $G$. If not,  there must exist a vertex $z$ in $T$ such that, naming $z'$ its out-neighbour in $T$, $N^{+}(z') \cup N^{-}(z') \cup N^{+}(z) \setminus V(T)$ is not monochromatic. Let $c$ be the colour of the out-neighbour of $z$ not in $T$. We can then safely assign colour $c$ to $z'$ and then greedily extend the dicolouring to $T \setminus \{z\}$.  Now, since the two out-neighbours of $z$ are coloured $c$, we can safely assign the other colour to $z$ to obtain a proper $2$-dicolouring of $G$. 
\medskip

We can now assume that each vertex $u$ of $T$ is linked to $G - T$ via a digon. 
If there is a vertex $x$ in $G-T$ linked to all vertices of $T$, then $T$ has at most $k$ vertices and thus must be isomorphic to $\olra K_{k}$. Hence $T \cup \{x\}$ induces $\olra K_{k+1}$, a contradiction. 

So, there exist two distinct vertices $x,y$ in $G-T$ linked via a digon to two (distinct) vertices of $T$. 
Let $H = G - T$ to which is added arcs $xy$ and $yx$ (if not existing). Then $H$ does not contain any element of $\mathcal B_{k}$ (as $G$ does not contain an element of $\mc B_{k}$ less a digon or an arc) and thus, by minimality of $G$, $H$ is $k$-dicolourable. 
Thus, $G - T$ admits a $k$-dicolouring in which $x$ and $y$ receive distinct colours. We can easily extend this $k$-dicolouring to a $k$-dicolouring of $G$ since each vertex of $T$ has a set of $k-1$ available colours and some pair  of vertices  in $T$ (the neighbours of $x$ and $y$) get distinct sets.
\end{proof}


\section{$k$-trees}

The proof presented in this section is an adaptation of a proof of Tverberg~\cite{T83}, see also section 4 of \cite{CR14}. 

A digraph $G$ is a \emph{direct composition} of digraphs $G_{1}$ and $G_{2}$ on vertices $v_{1} \in V(G_{1})$ and $v_{2} \in V(G_{2})$ if it can be obtained from the disjoint union of $G_1$ and $G_2$ by adding exactly one arc between $v_1$ and $v_2$ (either $v_1v_2$ or $v_2v_1$). 
A digraph $G$ is a \emph{cyclic composition} of digraphs $G_{1}, \dots, G_{\ell}$ ($\ell \geq 2$) on vertices $v_{1} \in V(G_{1}), \dots, v_{\ell} \in V(G_{\ell})$ if it can be obtained from the disjoint union of the $G_{i}$ by adding  the arcs $v_iv_{i+1}$ for $i=1, \dots, \ell- 1$ and $v_{\ell}v_1$

A digraph $G$ is a \emph{$k$-tree} if $\dmax(G) \leq k$ and it can be constructed as follows:
\begin{itemize}
    \item the digraphs in $\mathcal B_{k-1}$ are $k$-trees;
    \item a  direct or cyclic composition of $k$-trees is a $k$-tree;
\end{itemize}


Let $G$ be a digraph. 
A \emph{direct $k$-leaf} of $G$ is an induced subdigraph $T$ of $G$ such that $T$ belongs to $\mathcal B_{k-1}$ and $G$ is a direct composition of $T$ and $G-T$. 
If $G$ cannot be obtained from a cyclic composition of members of $\mathcal B_{k-1}$, an induced subdigraph $T$ of $G$ is a  \emph{cyclic $k$-leaf} of $G$ if $T$ can be obtained from $\ell \geq 1$ disjoint  $T_1, \dots, T_{\ell}$ belonging to $\mathcal B_{k-1}$ by adding $\ell -1 $ arcs $v_iv_{i+1}$ for $i=1, \dots, \ell-1$ where $v_i \in V(T_i)$, and $G$ is a cyclic composition of $G - T$ and $T_1, \dots, T_{\ell}$. See Figure~\ref{fig_tree}. 

A \emph{$k$-leaf} of $G$ is either a direct $k$-leaf or a cyclic $k$-leaf of $G$, or $G$ itself if $G$ is a member of $\mathcal B_{k-1}$ or $G$ is obtained from a cyclic composition of members of $\mathcal B_{k-1}$.
Observe that two distinct $k$-leaves of a digraph $G$ are always vertex disjoint and that a $k$-tree has at least two $k$-leaves except if it is a member of $\mathcal B_{k-1}$ or if  it can be obtained by a cyclic composition of members of $\mathcal B_{k-1}$. 

A \emph{$k$-path} is a digraph obtained by taking the disjoint union of $l \geq 2$ members $T_1, \dots, T_{\ell}$ of $\mathcal B_{k-1}$ and adding arcs $v_iv_{i+1}$ for $i=1, \dots, \ell-1$ where $v_i \in V(T_i)$.


    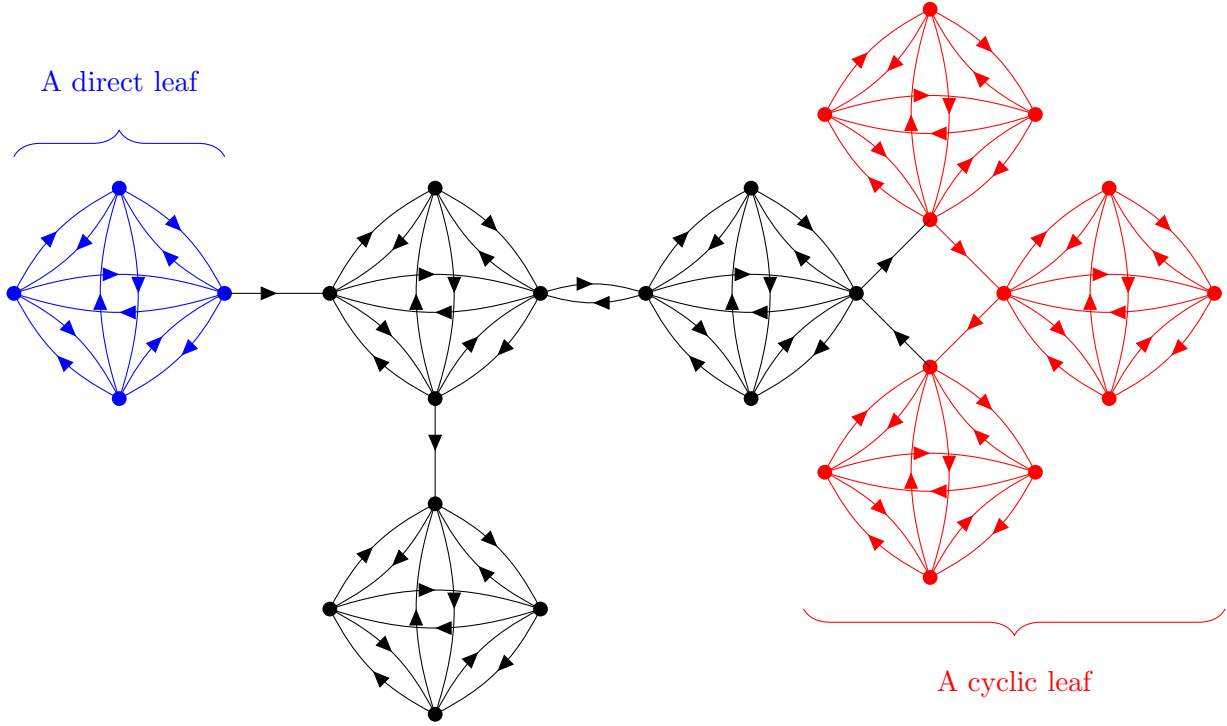
\begin{figure}[ht]
    
    \centering
    
	\begin{tikzpicture}[line cap=round,line join=round,>=triangle 45, scale=1.4] 
	
        \begin{scope}[xshift=0cm,yshift=0cm,scale=1]
            \draw[blue] [->-] (-1,0) to[bend left = 18] (0,-1);
            \draw[blue] [->-] (0,-1) to[bend left = 18] (-1,0);
            \draw[blue] [->-] (-1,0) to[bend left = 18] (0,1);
            \draw[blue] [->-] (0,1) to[bend left = 18] (-1,0);
            \draw[blue] [->-] (-1,0) to[bend left = 18] (1,0);
            \draw[blue] [->-] (1,0) to[bend left = 18] (-1,0);
            \draw[blue] [->-] (0,-1) to[bend left = 18] (0,1);
            \draw[blue] [->-] (0,1) to[bend left = 18] (0,-1);
            \draw[blue] [->-] (0,-1) to[bend left = 18] (1,0);
            \draw[blue] [->-] (1,0) to[bend left = 18] (0,-1);
            \draw[blue] [->-] (0,1) to[bend left = 18] (1,0);
            \draw[blue] [->-] (1,0) to[bend left = 18] (0,1);
            \fill[blue] (-1,0) circle (2pt);
            \fill[blue] (0,-1) circle (2pt);
            \fill[blue] (0,1) circle (2pt);
            \fill[blue] (1,0) circle (2pt);
        \end{scope}

        \begin{scope}[xshift=3cm,yshift=0cm,scale=1]
            \draw [->-] (-1,0) to[bend left = 18] (0,-1);
            \draw [->-] (0,-1) to[bend left = 18] (-1,0);
            \draw [->-] (-1,0) to[bend left = 18] (0,1);
            \draw [->-] (0,1) to[bend left = 18] (-1,0);
            \draw [->-] (-1,0) to[bend left = 18] (1,0);
            \draw [->-] (1,0) to[bend left = 18] (-1,0);
            \draw [->-] (0,-1) to[bend left = 18] (0,1);
            \draw [->-] (0,1) to[bend left = 18] (0,-1);
            \draw [->-] (0,-1) to[bend left = 18] (1,0);
            \draw [->-] (1,0) to[bend left = 18] (0,-1);
            \draw [->-] (0,1) to[bend left = 18] (1,0);
            \draw [->-] (1,0) to[bend left = 18] (0,1);
            \fill (-1,0) circle (2pt);
            \fill (0,-1) circle (2pt);
            \fill (0,1) circle (2pt);
            \fill (1,0) circle (2pt);
        \end{scope}

        \begin{scope}[xshift=3cm,yshift=-3cm,scale=1]
            \draw [->-] (-1,0) to[bend left = 18] (0,-1);
            \draw [->-] (0,-1) to[bend left = 18] (-1,0);
            \draw [->-] (-1,0) to[bend left = 18] (0,1);
            \draw [->-] (0,1) to[bend left = 18] (-1,0);
            \draw [->-] (-1,0) to[bend left = 18] (1,0);
            \draw [->-] (1,0) to[bend left = 18] (-1,0);
            \draw [->-] (0,-1) to[bend left = 18] (0,1);
            \draw [->-] (0,1) to[bend left = 18] (0,-1);
            \draw [->-] (0,-1) to[bend left = 18] (1,0);
            \draw [->-] (1,0) to[bend left = 18] (0,-1);
            \draw [->-] (0,1) to[bend left = 18] (1,0);
            \draw [->-] (1,0) to[bend left = 18] (0,1);
            \fill (-1,0) circle (2pt);
            \fill (0,-1) circle (2pt);
            \fill (0,1) circle (2pt);
            \fill (1,0) circle (2pt);
        \end{scope}

        \begin{scope}[xshift=6cm,yshift=0cm,scale=1]
            \draw [->-] (-1,0) to[bend left = 18] (0,-1);
            \draw [->-] (0,-1) to[bend left = 18] (-1,0);
            \draw [->-] (-1,0) to[bend left = 18] (0,1);
            \draw [->-] (0,1) to[bend left = 18] (-1,0);
            \draw [->-] (-1,0) to[bend left = 18] (1,0);
            \draw [->-] (1,0) to[bend left = 18] (-1,0);
            \draw [->-] (0,-1) to[bend left = 18] (0,1);
            \draw [->-] (0,1) to[bend left = 18] (0,-1);
            \draw [->-] (0,-1) to[bend left = 18] (1,0);
            \draw [->-] (1,0) to[bend left = 18] (0,-1);
            \draw [->-] (0,1) to[bend left = 18] (1,0);
            \draw [->-] (1,0) to[bend left = 18] (0,1);
            \fill (-1,0) circle (2pt);
            \fill (0,-1) circle (2pt);
            \fill (0,1) circle (2pt);
            \fill (1,0) circle (2pt);
        \end{scope}

        \begin{scope}[xshift=7.7cm,yshift=1.7cm,scale=1]
            \draw[red] [->-] (-1,0) to[bend left = 18] (0,-1);
            \draw[red] [->-] (0,-1) to[bend left = 18] (-1,0);
            \draw[red] [->-] (-1,0) to[bend left = 18] (0,1);
            \draw[red] [->-] (0,1) to[bend left = 18] (-1,0);
            \draw[red] [->-] (-1,0) to[bend left = 18] (1,0);
            \draw[red] [->-] (1,0) to[bend left = 18] (-1,0);
            \draw[red] [->-] (0,-1) to[bend left = 18] (0,1);
            \draw[red] [->-] (0,1) to[bend left = 18] (0,-1);
            \draw[red] [->-] (0,-1) to[bend left = 18] (1,0);
            \draw[red] [->-] (1,0) to[bend left = 18] (0,-1);
            \draw[red] [->-] (0,1) to[bend left = 18] (1,0);
            \draw[red] [->-] (1,0) to[bend left = 18] (0,1);
            \fill[red] (-1,0) circle (2pt);
            \fill[red] (0,-1) circle (2pt);
            \fill[red] (0,1) circle (2pt);
            \fill[red] (1,0) circle (2pt);
        \end{scope}

        \begin{scope}[xshift=7.7cm,yshift=-1.7cm,scale=1]
            \draw[red] [->-] (-1,0) to[bend left = 18] (0,-1);
            \draw[red] [->-] (0,-1) to[bend left = 18] (-1,0);
            \draw[red] [->-] (-1,0) to[bend left = 18] (0,1);
            \draw[red] [->-] (0,1) to[bend left = 18] (-1,0);
            \draw[red] [->-] (-1,0) to[bend left = 18] (1,0);
            \draw[red] [->-] (1,0) to[bend left = 18] (-1,0);
            \draw[red] [->-] (0,-1) to[bend left = 18] (0,1);
            \draw[red] [->-] (0,1) to[bend left = 18] (0,-1);
            \draw[red] [->-] (0,-1) to[bend left = 18] (1,0);
            \draw[red] [->-] (1,0) to[bend left = 18] (0,-1);
            \draw[red] [->-] (0,1) to[bend left = 18] (1,0);
            \draw[red] [->-] (1,0) to[bend left = 18] (0,1);
            \fill[red] (-1,0) circle (2pt);
            \fill[red] (0,-1) circle (2pt);
            \fill[red] (0,1) circle (2pt);
            \fill[red] (1,0) circle (2pt);
        \end{scope}

        \begin{scope}[xshift=9.4cm,yshift=0cm,scale=1]
            \draw[red] [->-] (-1,0) to[bend left = 18] (0,-1);
            \draw[red] [->-] (0,-1) to[bend left = 18] (-1,0);
            \draw[red] [->-] (-1,0) to[bend left = 18] (0,1);
            \draw[red] [->-] (0,1) to[bend left = 18] (-1,0);
            \draw[red] [->-] (-1,0) to[bend left = 18] (1,0);
            \draw[red] [->-] (1,0) to[bend left = 18] (-1,0);
            \draw[red] [->-] (0,-1) to[bend left = 18] (0,1);
            \draw[red] [->-] (0,1) to[bend left = 18] (0,-1);
            \draw[red] [->-] (0,-1) to[bend left = 18] (1,0);
            \draw[red] [->-] (1,0) to[bend left = 18] (0,-1);
            \draw[red] [->-] (0,1) to[bend left = 18] (1,0);
            \draw[red] [->-] (1,0) to[bend left = 18] (0,1);
            \fill[red] (-1,0) circle (2pt);
            \fill[red] (0,-1) circle (2pt);
            \fill[red] (0,1) circle (2pt);
            \fill[red] (1,0) circle (2pt);
        \end{scope}

        \draw [->-] (1,0) -- (2,0);
        \draw [->-] (3,-1) -- (3,-2);
        \draw [->-] (4,0) to[bend left = 18] (5,0);
        \draw [->-] (5,0) to[bend left = 18] (4,0);
        \draw [->-] (7,0) -- (7.7,0.7);
        \draw [red][->-] (7.7,0.7) -- (8.4,0);
        \draw [red][->-] (8.4,0) -- (7.7,-0.7);
        \draw [->-] (7.7,-0.7) -- (7,0);
        
    \draw[red] [decorate,decoration={brace,mirror,amplitude=10pt}]
    (6.5,-3) -- (10.5,-3) node [midway,yshift=-1cm] {\textcolor{red}{A cyclic leaf}};

    \draw[blue] [decorate,decoration={brace,amplitude=10pt}]
    (-1,1.3) -- (1,1.3) node [midway,yshift=1cm] {\textcolor{blue}{A direct leaf}};

	\end{tikzpicture}
    \caption{A 4-tree} \label{fig_tree}
    \end{figure}

The following easy observation will be useful during the proof. 
\begin{observation}\label{obs_1}
Let $G$ be a $k$-tree. 
Then all vertices of $G$ have mindegree at least $k-1$. Moreover, $G$ has at least $k+1$ vertices of mindegree $k-1$, except if $G = \olra{K}_{k}$ or if it is a symmetric path of odd length (and thus $k=2$). 
\end{observation}

The main ingredient of the proof is the following lemma:

\begin{lemma}\label{lem:ktree}
Let $G$ be a connected digraph and $k = \Delta_{max}(G) \geq 2$. 
Then either $G$ is a member of $\mathcal B_k$, or $G$ is a $k$-tree, or there exists $v \in V(G)$ such that $d_{max}(v) = k$ and no connected component of $G - \{v\}$ is a $k$-tree. 
\end{lemma}

\begin{proof}
Let $G$ be a digraph with $\dmax(G) = k$ and assume that $G$ is not a member of $\mathcal B_k$ nor a $k$-tree. 

\begin{claim}\label{clm:no_leaf}
$G$ has no $k$-leaf.
\end{claim}

\begin{proofclaim}
Assume first that $G$ has a direct $k$-leaf $T$, and let $v$ be the unique vertex of $T$ having a neighbour outside of $T$. Recall that $T$ belongs to $\mathcal B_{k-1}$ by definition of a direct $k$-leaf. 
Then $d_{max}(v) = k$ and $G-\{v\}$  has two connected components, $T-\{v\}$ and $G -T$. $G-T$ is not a $k$-tree otherwise $G$ is too, and $T-\{v\}$ is clearly not a $k$-tree, so we are done. 

Assume now that $G$ has a cyclic $k$-leaf $T$ made of $\ell$ members $T_1, \dots, T_{\ell}$ of $\mathcal B_{k-1}$ and let $v_1, \dots, v_{\ell}$ be as in the definition of cyclic $k$-leaf. Then $d_{max}(v_1) = k$ and $G - \{v_1\}$ has two connected components, $T_1-\{v_1\}$ and $G-T_1$. As in the previous case, none of them is a $k$-tree. 
\end{proofclaim}

We say that a vertex $v$ of $G$ is \textit{special} if it is contained in an induced subdigraph  of $G$ isomorphic to a member of $\mc B_{k-1}$ and $d_{max}(v) = k$. 
For each special vertex $x$, choose arbitrarily an induced subdigraph of $G$ isomorphic to a member of $\mathcal B_{k-1}$ that we name $T_x$. Moreover, we name $H_x$ the connected component of $G-x$ containing $T_x - x$. Note that in the case where $G- x$ is connected, we have $H_{x} = G- x$. 

If no induced subdigraph of $G$ is isomorphic to a member of $\mathcal B_{k-1}$, then any vertex $v$ with maxdegree $k$ is such that no component of $G - \{v\}$ is a $k$-tree. 
Moreover, if $G$ has an induced subdigraph $H$ isomorphic to a member of $\mc B_{k-1}$, then at least one of its vertices must have a maxdegree equal to $k$, otherwise $G = H$ is a $k$-tree, a contradiction. 
Hence, $G$ must contain some special vertices, and every subdigraph of $G$ isomorphic to a member of $\mathcal B_{k-1}$ contains a special vertex. 
\medskip

Assume there exists a special vertex $v$ such that $H_v$ is not a $k$-tree. 
If $G-v$ is connected, then $v$ is such that $d_{max}(v) = k$ and no component of $G - \{ v \}$ is a $k$-tree. 
So we can assume $G-v$ is not connected. 

Assume first $v$ has only one neighbour $a$ in $G-H_v$. Let $G_a$ be the connected component of $G-v$ containing $a$. We may assume  $G_a$ is a $k$-tree, otherwise $v$ is such that $d_{max}(v) = k$ and no component of $G - \{ v \}$ is a $k$-tree.  
If $G_a$ is isomorphic to a member of $\mathcal B_{k-1}$, then $G_a$ is a $k$-leaf of $G$ (direct of cyclic depending if $a$ and $v$ are linked by a single arc of a digon), if $G_a$ is a cyclic composition of members of $\mathcal B_{k-1}$, then $G$ contains a cyclic $k$-leaf, and otherwise $G_a$ has at least two $k$-leaves, one of the two does not contain $a$ and is thus a $k$-leaf of $G$. Each case contradicts (\ref{clm:no_leaf}).

So $v$ has at least two neighbours $a$ and $b$ in $G-H_v$, and $a \neq b$. 
If $a$ and $b$ are in two distinct connected component $G_a$ and $G_b$ of $G-v$, then one of $G_a$ or $G_b$ must be a $k$-tree, for otherwise $v$ is such that $d_{max}(v) = k$ and no component of $G - \{ v \}$ is a $k$-tree, and we find a $k$-leaf as in the previous case.

So we may assume that $G-H_v$ is connected. 
Moreover, $G-H_v$ must be a $k$-tree, for otherwise $v$ is such that $d_{max}(v) = k$ and no component of $G - \{ v \}$ is a $k$-tree.
If $G-H_v$ has a $k$-leaf disjoint from $\{a,b\}$, then it is a $k$-leaf of $G$, a contradiction to (\ref{clm:no_leaf}). So $G-H_v$ is isomorphic to a member of $\mathcal B_{k-1}$ or is a cyclic composition of members of $\mathcal B_{k-1}$ or has exactly two leaves, $T_a$ and $T_b$ containing respectively $a$ and $b$. 

If $G-H_v$ is a member of $\mathcal B_{k-1}$, then $G-a$ is connected and is not a $k$-tree, so $a$ is such that $d_{max}(a) = k$ and no component of $G - \{ a \}$ is a $k$-tree 
If $G-H_v$ is a cyclic composition of members of $\mathcal B_{k-1}$, then $a$ cannot be a cutvertex of $G-H_v$ (otherwise $d_{max}(a)>k$), and thus $G-a$ is connected and is not a $k$-tree, so again $a$ is such that $d_{max}(a) = k$ and no component of $G - \{ a \}$ is a $k$-tree. 

So $H_v$ has exactly two leaves $T_a$, $T_b$ as explained above. 
Observe that the only vertex of $T_a$ with maxdegree $k$ in $G$ is $a$, for otherwise $G-a$ is connected and is not a $k$-tree, so $a$ satisfies the theorem. 
Same holds for $T_b$ and $b$. 
Let $T$ be an induced subdigraph of $G-H_v$ isomorphic to a member of $\mathcal B_{k-1}$ that does not contain $a$ nor $b$. 
If $T$ has at least $3$ vertices of maxdegree $k$, then $G-H_v$ contains a $k$-leaf disjoint from $\{a,b\}$, a contradiction to (\ref{clm:no_leaf}). 
If $T$ has exactly two vertices of maxdegree $k$, then deleting one leads to a connected digraph which is not a $k$-tree and we are done. 
So we may assume that each subdigraph of $G-H_v$ isomorphic to a member of $\mathcal B_{k-1}$ contains exactly one vertex of maxdegree $k$. 
It implies that $G-H_v$ is a $k$-path and that $G$ is a cyclic composition of members of $\mathcal B_{k-1}$ and thus a $k$-tree, a contradiction. 
\medskip

We may now assume that for every special vertex $v$, $H_v$ is a $k$-tree. 
Let $x$ be a special vertex and assume without loss of generality that $d^-(x) = k$. 
Let $S$ the set of vertices in $T_x$ with in-degree $k$. 
If $T_x -S$ is non-empty, there must exists an arc $st$ where $s \in S$ and $t \in T_x - S$ (because $T_x$ is strongly connected).  
Since $H_s$ is a $k$-tree, $t$ must have in-degree at least $k-1$ in $G-s$, and thus has in-degree $k$ in $G$, a contradiction. So every vertex of $T_x$ has in-degree $k$. 
Let $y$ be an in-neighbour of $x$ in $T_x$. 
As $H_x$ is a $k$-tree, $y$ has out-degree at least $k-1$ in $H_x$, and thus has out-degree $k$ in $G$. Now, by the same reasoning as above, we get that every vertex of $T_x$ has out-degree $k$. 
This proves that for every special vertex $v$, every vertex $u$ in $T_v$ has in- and out-degree $k$.  

Let $x$ be a special vertex. We know that $H_x$ is a $k$-tree. So every vertex of $H_x$ is contained in a subdigraph isomorphic to a member of $\mc B_{k-1}$, and thus has in- and out-degree $k$ in $G$. 
Hence, every vertex of $H_x$ has in- and out- degree $k$ in $H_x$ except the neighbours of $v$. So $H_v$ has at most $k$ vertices of mindegree $k-1$. 
If $k \geq 3$, it implies that $H_x$ is isomorphic to $\olra{K_k}$ and thus $G = \olra K_{k+1}$, a contradiction. And if $k=2$, it implies that $H_x$ is a symmetric path of odd length (obtained by doing a sequence of cyclic composition of digons) and thus $G$ is a symmetric cycle of odd length, a contradiction.

\end{proof}

\begin{theorem}
Let $G$ be a connected digraph with $\Delta_{max}(G) = k$. Then $\dic(G) \leq k+1$ and equality occurs if and only if $G$ is a member of $\mathcal B_k$. 
\end{theorem}

\begin{proof} 
We proceed by induction on $k$, so we may assume $k \geq 2$. 
If $G$ is a member of $\mathcal B_k$, then we are done.  If $G$ is a $k$-tree, then it is $k$-dicolourable because members of $\mathcal B_{k-1}$ are $k$-dicolourable, and compositions preserve $k$-dicolourability. 

So, by Lemma \ref{lem:ktree}, $G$ has a vertex $v_1$ with $d_{max}(v_1) = k$ and such that no connected component of $G-\{v_1\}$ is a $k$-tree. Let $G_2, \dots, G_r$ be the connected components of $G - \{v_1\}$. 
Observe first that each $G_i$ has a vertex with mindegree at most $k-1$, so it is not a member of $\mathcal B_k$. 
For each $G_i$, either $\dmax(G_i) \leq k-1 $ and since $G_i$ is not a $k$-tree, it is $k-1$-dicolourable by induction, or, by Lemma~\ref{lem:ktree}, $G_i$ contains a vertex $v_i$ such that the maxdegree of $v_i$ in $G_i$ equal $k$ and no connected component of $G_i - \{v_i\}$ is a $k$-tree. In the latter case, we choose such a vertex $v_i$, and continue this procedure on the connected components of $G_i \setminus \{v_i\}$ and so on. 

We obtain a set of ordered vertices $v_1, \dots, v_s$ (we apply the procedure level by level, putting an arbitrary order inside each level) such that $v_i$ has either no in-neighbour or no out-neighbour in $\{v_1, \dots, v_{i-1}\}$ (because maxdegree of $v_i$ in $G_i$ is $k = \dmax(G)$). So the digraph induced by $\{v_1, \dots, v_s\}$ is acyclic. 
Moreover, $G-\{v_1, \dots, v_s\}$ is made of vertex disjoint $(k-1)$-dicolourable induced subgraph of $G$. Hence, $G$ is $k$-dicolourable.  
\end{proof}


\section{Partitioned dicolouring}
In this section, we adapt a proof of Brooks' Theorem based on a specific partition of the vertices introduced by Lov\'asz in~\cite{L66}.  See section 5 of \cite{CR14} for the undirected version of the proof as well as a short history of the involved methods. Same kind of methods has been recently used in~\cite{BSS21} to prove a generalisation of the directed Brooks' Theorem. 

Let $G = (V,A)$ be a digraph. 
We say that $G$ is \emph{$r$-special} if for every vertex $v \in V$, either $d_{min}(v) < r$ or $d_{min}(v) = d_{max}(v) = r$ (note that last equality is equivalent to $d^+(v)= d^-(v) = r$). 
Let $r_1$ and $r_2$ be two positive integers. A partition $\mc P = (V_{1}, V_{2})$ of $V(G)$ is \emph{($r_1, r_2$)-normal} if it minimizes $r_{2}|A(G[V_1])| + r_{1}|A(G[V_2)]|$.

Next observation is used frequently in the proof and is a basic property of $(r_1,r_2)$-normal partition. 
\begin{observation}\label{obs:special}
Let $G$ be a digraph. If $\mc P$ is a $(r_1, r_2)$-normal partition of $G$ with $r_1 + r_2 \geq \dmax(G) \geq 1$, then $G[V_{1}]$ is $r_{1}$-special and $G[V_{2}]$ is $r_{2}$-special. 
\end{observation}

\begin{proof}
Assume for contradiction and without loss of generality that $G[V_1]$ is not $r_1$-special. Then there is $v_1 \in V_1$ such that $d_{min}(v_1) \geq r_1$ and $d_{max}(v_1) \geq r_1 + 1$ in $G[V_1]$. Assume without loss of generality that $d^-_{G[V_1]}(v_1) \geq r_1$ and $d^+_{G[V_1]}(v_1) \geq r_1 + 1$. 

Set $V'_1 = V_1 \setminus \{v_1\}$ and $V'_2 = V_2 \cup \{v_1\}$ and let us prove that the partition $(V'_1, V'_2)$ contradicts the fact that $(V_1, V_2)$ is $(r_1, r_2)$-normal. 
Since $r_1 + r_2 \geq \dmax(G)$, we have that $d^+_{G[V'_2]}(v_1) \leq r_2 - 1$ and $d^-_{G[V'_2]}(v_1) \leq r_2$. Hence:
$$(r_{2}|A(G[V_1])| + r_{1}|A(G[V_2)]|) - (r_{2}|A(G[V'_1])| + r_{1}|A(G[V'_2)]|) \leq -(2r_1+1)r_2 +r_1(2r_2-1) = -r_1-r_2<0$$ a contradiction. 

\end{proof}

Let $G$ be a digraph, and  $\mc P$  a $(r_1,r_2)$-normal partition of $G$ with $r_1 + r_2 \geq \dmax(G)$. 
We define the \emph{$\mc P$-components} of $G$ as the connected components of $G[V_{1}]$ and $G[V_{2}]$. 
A $\mc P$-component is an \emph{obstruction} if it is a member of $\mathcal B_{r_{1}}$ in $G[V_{1}]$ or a member of $\mathcal B_{r_{2}}$ in $G[V_{2}]$. 
A path $v_1 \dots v_k$  in the underlying graph of $G$ is \emph{$\mc P$-acceptable} if $v_{1}$ is in an obstruction and  vertices of $\mc P$ are in pairwise distinct $\mc P$-components.
We say that a  $\mc P$-acceptable path is \emph{maximal} if every neighbour of $v_{k}$ is in the same $\mc P$-component as some vertex in the path. 
Given a partition $\mc P$, to \emph{move a vertex $u$} is to move it to the other part of $\mc P$. 

Observation~\ref{obs:special} together with the fact that digraphs in $\mathcal B_k$ are $k$-regular easily implies the following facts that will be used routinely during the proof:
\begin{itemize}
    \item If a $\mc P$-component contains an obstruction, then the obstruction is the whole $\mc P$-component. 
    \item If a vertex $u$ is in an obstruction, then the partition created by moving $u$ is again $(r_1, r_2)$-normal.
\end{itemize}

\begin{lemma} \label{lem:partition}
Let $k \geq 2$. 
Let $G = (V, A)$ be a $k$-regular connected digraph not in $\mathcal B_{k}$ and let $r_{1}, r_{2} \geq 1$ such that $r_{1} + r_{2} = k$. There exists a $(r_1, r_2)$-normal partition $(V_{1},V_{2})$ such that, for $i \in \{1,2\}$, $G[V_{i}]$ is $r_i$-special and has no obstruction.
\end{lemma}

\begin{proof}
By Observation~\ref{obs:special}, for every $(r_1,r_2)$-normal partition $(V_1, V_2)$, $G[V_i]$ is $r_i$-special for $i=1, 2$. 

Suppose that the lemma is false and let $G$ be a counterexample.  Among the $(r_1, r_2)$-normal partitions of $G$ with the minimum number of obstructions, let $\mc P = (V_1, V_2)$ be one with the shortest maximal $\mc P$-acceptable path $v_{1}\dots v_{\ell}$. 
We refer to the minimality of the number of obstructions by saying ``by minimality of $\mc P$", and to the minimality of the $\mc P$-acceptable path by saying ``by minimality of $\ell$".

Throughout the proof, we often move some vertex $u$ that belongs to an  obstruction $A$. 
Since this destroys $A$ and results in a $(r_1,r_2)$-normal partition, the minimality of $\mc P$  implies that the move creates a new obstruction and thus the obtained partition has the same number of obstructions as $\mc P$. 
Moreover, this new obstruction  contains $u$  and  the neighbours of $u$ in the other part. This implies that the neighbours of $u$ in the other part are contained in a single $\mc P$-component $C$ (because obstructions do not have cut-vertex),   and that $C \cup u$ is an obstruction. Finally, note that an obstruction containing a digon is a symmetric digraph. 
These facts are constantly used in the proof. 

Let $A$ and $B$ be the $\mc P$-components containing $v_{1}$ and $v_{\ell}$ respectively. 
Let $X=N_{A}(v_{\ell})$. 
\medskip 

Assume  $X= \emptyset$. 
Moving $v_{1}$ creates a new $(r_1, r_2)$-normal partition $\mc P'$. Since $v_{1}$ is adjacent to $v_{2}$, the new obstruction contains $v_{2}$. Moreover, $A \setminus v_{1}$ is not an obstruction. So $v_{2}, v_{3} \dots v_{\ell}$ is a maximal $\mc P'$-acceptable path, violating the minimality of $\ell$. Hence $|X| \geq 1$.
\medskip 

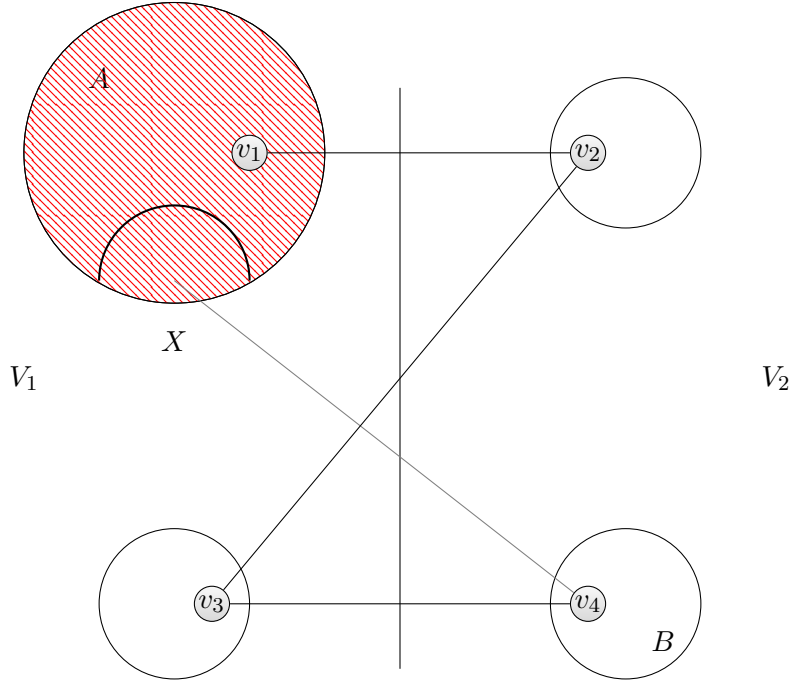
\begin{figure}[ht]
    \centering
    \begin{tikzpicture}
   
        \node (V1) at (-5,0) {$V_1$};
        \node (V2) at (5,0) {$V_2$};
    
        \node (A) at (-4,4) {$A$};
        \node (B) at (3.5,-3.5) {$B$};
        \node (X) at (-3,0.5) {$X$};
        
        \node (down) at (0,-4) {};
        \node (up) at (0,4) {};
        \draw (down) -- (up);
        
        \node (ul) at (-3,3) {};
        \node (ur) at (3,3) {};
        \node (dl) at (-3,-3) {};
        \node (dr) at (3,-3) {};

        \draw (ul) circle (2);
        \filldraw[pattern=north west lines, pattern color=red] (ul) circle (2);
        
        \draw (ur) circle (1);
        \draw (dl) circle (1);
        
        \draw (dr) circle (1);
    
        \centerarc[thick](-3,1.3)(0:180:1);

        \node (ul_v) at (-2,3) [vertex] {$v_1$};
        \node (ur_v) at (2.5,3) [vertex] {$v_2$};
        \node (dl_v) at (-2.5,-3) [vertex] {$v_3$};
        \node (dr_v) at (2.5,-3) [vertex] {$v_4$};

        \draw[gray] (dr_v) -- (-3, 1.3);
        
        \draw (ul_v) -- (ur_v);
        \draw (ur_v) -- (dl_v);
        \draw (dl_v) -- (dr_v);
    
    \end{tikzpicture}

    \caption{A partition $\mc P = (V_1, V_2)$, a maximal $\mc P$-acceptable path along with $X$, $A$ and $B$ as defined in the proof of Lemma~\ref{lem:partition}. Red colour indicates an obstruction. $G[V_1]$ is $r_1$-normal and $G[V_2]$ is $r_2$-normal.} 
    \label{fig:my_label}
\end{figure}

Assume now that $|X| \geq 2$ and let $x_1$, $x_2$ be two vertices in $X$. 
Let us first prove that $G[X \cup v_{\ell}$ is a symmetric complete graph. 
Assume that $x_1v_{\ell} \in A(G)$ (the case $v_{\ell}x_1 \in A(G)$ is similar).
As explained above (in the second paragraph of the proof), $B \cup x_1$ and $B \cup x_2$ are obstructions, which implies that $x_2v_{\ell} \in A(G)$. This is because $x_1v_{\ell}$ is an arc and obstructions are regular. 
By moving $x_1$ and then $v_{\ell}$, we get that $(A \sm x_1) \cup v_{\ell}$ is an obstruction, so $x_2x_1 \in A(G)$ (again because obstructions are regular). 
Similarly, $(A \sm x_2) \cup v_{\ell}$ is an obstruction and thus $x_1x_2 \in A(G)$. So $x_1$ and $x_2$ are linked by a digon, which implies that $v_{\ell}$ is linked to $x_1$ and $x_2$ by digons (this is again because obstructions are regular and $(A \sm x_1) \cup v_{\ell}$ and $(A \sm x_2) \cup v_{\ell}$ are obstructions).  
We deduce that $G[X \cup v_{\ell}]$ is a symmetric complete graph  

Let us now prove that $G[A\cup v_{\ell}]$ is a symmetric complete digraph. 
Since $A$ is an obstruction and $x_1$ and $x_2$ are linked by a digon, $A$  induces a symmetric digraph. 
If $A = X$ we are done, so we may assume that $A$ has at least three vertices. 
Since $(A \sm x_1) \cup v_{\ell}$ is an obstruction, $v_{\ell}$ has at least two neighbours in  $A \sm x_1$ and thus $|X| \geq 3$. 
Since $X$ induces a complete symmetric digraph, $A$ contains a symmetric triangle and thus must be a symmetric complete digraph. This implies that $G[A\cup v_{\ell}]$ is a symmetric complete digraph as announced. 

Let us now prove that $B$ and $A \cup B$ also induce a symmetric complete graph. 
Since $G[A \cup v_{\ell}]$ induces a complete symmetric digraph, for every $a \in A$, $B \cup a$ is an obstruction. This implies that each vertex of $A$ share the same neighbourhood in $B$ and that $B$ induces a symmetric digraph.  
If $B = \{v_{\ell}\}$ we are done, so $B$ has at least two vertices. Let $a \in A$. 
Since $B \cup a$ is an obstruction, $B \cup a$ contains a symmetric triangle, and thus $B$ is a complete symmetric digraph. Finally, it implies that for every $a \in A$, $B \cup \{a\} \sm \{v_{\ell}\}$ induces a complete symmetric digraph, and so $A \cup B$ induces a complete symmetric digraph. 

All together, this proves that $G[A] = \overleftrightarrow K_{r_1+1}$, $G[B] = \overleftrightarrow K_{r_2}$ (because for every $a \in A$, $G[A]$ is an obstruction i.e. is a member of $\mathcal B_{r_2}$, and is a complete symmetric digraph). So $A \cup B$ induces $\overleftrightarrow K_{r_1 + r_2 + 1} = \overleftrightarrow K_{k+1}$ and since $G$ is $k$-regular, $G = \overleftrightarrow K_{k+1}$, a contradiction with the hypothesis that $G$ is not a member of $\mc B_{k}$. 
\medskip 

We may assume from now on that $|X| = 1$. 
Assume first that $X= \{v_1\}$. 
Moving $v_1$ creates an obstruction containing both $v_2$ and $v_{\ell}$, so $\ell = 2$. 
Since the path $v_1v_2$ is a maximal $\mc P$-acceptable path, $v_2 = v_{\ell}$ has no neighbour in the other part besides $v_1$. 
Hence, after moving $v_1$ and $v_2$, $v_2$ is the only vertex in its component, and thus cannot be in an obstruction, a contradiction. 

So instead $X = \{x\}$ and $x \neq v_1$. 
Let us prove that $A = \{x,v_1\}$, that $A$ induces a digon, and that $v_{\ell}$ and $x$ are linked by a digon.  
In order to do so, move each $v_1, v_2, \dots, v_{\ell}$ in turns.  
Moving $v_1$ destroys  $A$ and thus  creates a new obstruction containing $v_2$.  
For $1 \leq i \leq \ell-2$, moving $v_i$ creates a new obstruction containing $v_{i+1}$, which in turns is destroyed by the move of $v_{i+1}$, creating a new obstruction containing $v_{i+2}$. 
Finally, after the move of $v_{\ell}$, $v_{\ell}$ is in an obstruction containing $x$ and  since $|X| = 1$, this new obstruction only contains $x$ and $v_{\ell}$, and thus is a digon. This also implied that $A = \{v_1, x\}$ and thus induces a digon. Moreover, it implies that $r_1 = 1$.   

Moving $v_1$ creates an obstruction containing $v_2$. By minimality of $\ell$, in the new partition $\mc P'$ obtained after moving $v_1$, the path $v_2v_3 \dots v_{\ell}x$ is a maximal $\mc P'$-acceptable path. 
So the obstruction containing $v_2$ (the first obstruction of a maximal acceptable path) must be a $\overleftrightarrow K_2$ (for the same reason $A$ is a $\overleftrightarrow K_2$), so $v_1$ and $v_2$ are linked by a digon and $r_2 = 1$. 
Now, moving $v_1$ and then $v_2$, the same argument can be applied to the path $v_3 \dots v_{\ell}xv_1$ implying that $v_2$ is linked to $v_3$ by a digon. 
Similarly, each $v_i$ for $i=2, \dots, \ell-1$ 
is linked to $v_{i+1}$ by a digon. 
This implies that $G$ contains a symmetric cycle of odd length (namely $v_1v_2 \dots v_{\ell}xv_1$), and since $G$ is $k$-regular and we clearly have $r_1 = r_2 =1$, $G$ is equal to this symmetric odd cycle, a contradiction. 
\end{proof}

\begin{theorem}
A connected digraph $G$ has dichromatic number at most $\dmax(G) +1$ and equality occurs if and only it is a member of $\mc B_{\dmax(G)}$. 
\end{theorem}

\begin{proof}
We proceed by induction on $\Delta_{max}$. 
Let $G$ be a connected digraph with $\dmax(G) = k \geq 2$. As usual, we may assume that $G$ is $k$-regular. 
If $G$ is a member of $\mc B_k$, then we are done, so we may assume that it is not and we need to prove that $G$ is $k$-dicolourable. 
Hence, by Lemma~\ref{lem:partition}, there exists a $(1, k-1)$-normal partition $(V_1, V_2)$ such that, for $i=1, 2$,  $G[V_i]$ is $r_i$-special and has no obstruction. 
Set $G_i = G[V_i]$ for $i=1, 2$. An obstruction in $G_1$ is a directed cycle, so $G_1$ is acyclic. 
We are now going to prove that $G_2$ is $k-1$-dicolourable. 
Let $S \subseteq V(G_2)$ be the set of vertices with maxdegree $k$ in $G_2$. 
Hence, every vertex in $V(G_2) \setminus S$ has maxdegree $k-1$ (in $G_2$) and has no $\overleftrightarrow K_{k-1}$ (because $G_2$ has no obstruction) so, by minimality of $k$, $G_2 \setminus S$ is $(k-1)$-dicolourable. 
Since $G_2$ is $(k-1)$-special, vertices in $S$ have mindegree at most $k-2$ in $G_2$. Hence, we can greedily extend a $k-1$-dicolouring of $G_2 \setminus S$ to $G_2$. Using one more colour for $V_1$, we get a $k$-dicolouring of $G$. 
\end{proof}


\section{No Brooks' analogue for $\dmin$} \label{sec:dmin}

As explained in the introduction, every digraph $G$ can be dicoloured with $\dmin(G) + 1$ colours. 
In this section, we prove that given a digraph $G$, deciding if it is $\dmin(G)$-dicolourable is $NP$-complete. 
It is thus unlikely that digraphs satisfying $\dic(G) = \dmin(G)+1$ admit a simple characterization, contrary to the digraphs satisfying $\dic(G) = \dmax(G) +1$.  

It is known that for all $k \geq 2$, $k$-\textsc{dicolourability} is NP-complete~\cite{Bo04}, where $k$-\textsc{dicolourability} is the following problem:\\
\underline{Input}: A digraph $G$.\\
\underline{Question}: Is $G$ $k$-dicolourable ?

\begin{theorem}
For all $k \geq 2$, $k$-\textsc{dicolourability} is $NP$-complete even when restricted to digraph $G$ with $\dmin(G) = k$. 
\end{theorem}

\begin{proof}
Let $k \geq 2$ be a fixed integer. 
As is customary, membership to $NP$ is clear.
Given a digraph $G$, we are going to construct a digraph $G'$ such that $\dmin(G') \leq k$ and $G$ is $k$-dicolourable if and only if $G'$ is $k$-dicolourable.

Let $G=(V,A)$ be a digraph. We construct $G'$ as follows:
\begin{itemize}
    \item For every vertex $u$ of $G$, put $k+1$ vertices in $G'$ : $u^-$,  $u^+$, $u_1$, \dots, $u_{k-1}$.
    \item for each vertex $u$, $G'[\{u^-, u_1, \dots, u_{k-1}\}]$ and $G'[\{u^+, u_1, \dots, u_{k-1}\}]$ are complete symmetric digraphs, and $u^-u^+ \in A(G')$. 
    \item For every $uv \in A(G)$, $u^+v^- \in A(G')$. 
\end{itemize}

For every vertex $u \in V(G)$, we have $d_{G'}^-(u^+) =d_{G'}^+(u^-) = k$ and for $i=1, \dots, k-1$, $d^+_{G'}(u_i) = k$. Hence, $\dmin(G') \leq k$. 









\begin{claim} 
If $G$ is $k$-dicolourable, then $G'$ is too. 
\end{claim}

\begin{proofclaim}
Let $\phi$ be a $k$-dicolouring of $G$. For every vertex $u$, assign to $u^-$ and $u^+$ the colour $\phi(u)$, and the $k-1$ other colours to $\{u_1, \dots, u_{k-1}\}$. We claim this is a proper $k$-dicolouring of $G'$. 
Suppose it is not. Let $C$ be a monochromatic directed cycle in $G'$. It cannot use any vertex $u_i$ as these vertices have a colour distinct from all of their neighbours. Thus $C$ only uses arcs of the form $u^-u^+$ or $u^+v^-$ which easily implies the existence of a monochromatic directed cycle in $G$, a contradiction.  
\end{proofclaim}

\begin{claim}
If $G'$ is $k$-dicolourable, then $G$ is too. 
\end{claim}

\begin{proofclaim}

Let $\phi$ be a $k$-dicolouring of $G'$. 
For every vertex $u \in V(G)$, for $i=1, \dots, k-1$, the vertices $u_i$ receive pairwise distinct colours. So, $\phi(u^+)=\phi(u^-)$. Hence, for every vertex $u \in V(G)$, assigning the colour $\phi(u^+)$ to $u$ gives a valid $k$-dicolouring of $G$. 
\end{proofclaim}
\end{proof}

\section*{Acknowledgements} 
This research was partially supported by the french Agence Nationale de la Recherche under contract DAGDigDec (JCJC) ANR-21-CE48-0012, and by the group Casino/ENS Chair on Algorithmics and Machine Learning.

\end{document}